\documentclass[pdflatex,sn-mathphys-num]{sn-jnl}

\usepackage{graphicx}%
\usepackage{multirow}%
\usepackage{amsmath,amssymb,amsfonts}%
\usepackage{amsthm}%
\usepackage{mathrsfs}%
\usepackage[title]{appendix}%
\usepackage{xcolor}%
\usepackage{textcomp}%
\usepackage{manyfoot}%
\usepackage{booktabs}%
\usepackage{algorithm}%
\usepackage{algpseudocode}%
\usepackage{listings}%

\theoremstyle{thmstyleone}%
\newtheorem{theorem}{Theorem}
\theoremstyle{thmstyletwo}%

\theoremstyle{thmstylethree}%
\newtheorem{definition}{Definition}%

\begin{document}

\title[A Bayesian Approach to Low-Discrepancy Subset Selection]{A Bayesian Approach to Low-Discrepancy Subset Selection}


\author*[1]{\fnm{Nathan} \sur{Kirk}}\email{nmk7@st-andrews.ac.uk}

\affil*[1]{\orgdiv{School of Mathematics and Statistics}, \orgname{University of St Andrews}, \orgaddress{\street{Mathematical Institute}, \city{St Andrews}, \postcode{KY16 9SS}, \country{United Kingdom}}}


\abstract{Low-discrepancy designs play a central role in quasi–Monte Carlo methods and are increasingly influential in other domains such as machine learning, robotics and computer graphics, to name a few. In recent years, one such low-discrepancy construction method called \emph{subset selection} has received a lot of attention. Given a large population, one optimally selects a small low-discrepancy subset with respect to a discrepancy-based objective. Versions of this problem are known to be NP--hard. In this text, we establish, for the first time, that the subset selection problem with respect to kernel discrepancies is also NP--hard. Motivated by this intractability, we propose a Bayesian Optimization procedure for the subset selection problem utilizing the recent notion of deep embedding kernels. We demonstrate the performance of the BO algorithm to minimize discrepancy measures and note that the framework is broadly applicable any design criteria.}




\maketitle

\section{Introduction}
An important problem in statistics and numerical computation is to approximate a probability distribution of interest by a finite set of representative points. This idea underlies a widely used variant of the Monte Carlo (MC) method known as \emph{quasi--Monte Carlo} (QMC) \cite{HICKKIRKSOR2025,Lem09a}, which replaces the naive random sampling of MC with so-called \emph{low-discrepancy} (LD) design sets whose elements are more evenly spread over the domain. Such designs find applications in numerical approximation, uncertainty quantification, Bayesian inference, computer vision, robotics, and scientific machine learning \cite{paulin2022,FanWin00,Keller2013a,HerSch20a,mishra21,longo21,Chahine2025}. Nowadays, many ``off-the-shelf'' LD designs are available and can typically be divided into two prominent families: digital nets and sequences \cite{DicPil10a} and lattice rules \cite{DicEtal22a}. Both rely heavily on ideas from number theory and abstract algebra. Complementary to these classical constructions, LD designs have also been generated computationally in recent years using modern optimization techniques, including neural network-based approaches \cite{ruschkirk24} and nonlinear programming methods \cite{cle25}.

In the classical QMC setting, discrepancy measures the irregularity of distribution of a point set 
$P_N := \{\mathbf{X}_i\}_{i=1}^N \subset [0,1]^d$ 
with respect to the uniform distribution on the unit hypercube. Perhaps the most popular discrepancy in the literature is the so-called $L_\infty$ \emph{star discrepancy}.

\begin{definition}[$L_\infty$ Star Discrepancy]
   Given a point set $P_N = \{\mathbf{X}_i\}_{i=1}^N \subset [0,1]^d$, the $L_\infty$ star discrepancy of $P_N$ is defined as
    \begin{equation}
        D_\infty^*(P_N) := \sup_{\boldsymbol{x}\in[0,1]^d} \left| \frac{\#(P_N \cap [0,\boldsymbol{x}))}{N} -\lambda([0,\boldsymbol{x})) \right|.\label{eq:star_disc}
        \end{equation}
    where $\lambda(\cdot)$ denotes the usual Lebesgue measure.
\end{definition}

This notion has been studied extensively over several decades and because of its nature---testing the uniformity of subintervals of the unit hypercube---it is often referred to as the geometric interpretation of discrepancy \cite{KuiNie74, matousek1999geometric}. On the other hand, discrepancy can also be viewed from a kernel methods perspective, an approach commonly attributed to \cite{hickernell1998}. In this paper, we consider a more general notion of kernel discrepancy than that introduced in \cite{hickernell1998}, namely the \emph{maximum mean discrepancy} \cite{gretton2006}, which encompasses several classical and modern kernel discrepancy settings within a unified framework.

\begin{definition}[Maximum Mean Discrepancy]
Let $P$ and $Q$ be Borel probability measures on a measurable space $\mathcal{X}$. 
Let $\mathcal{H}$ be a reproducing kernel Hilbert space (RKHS) with reproducing kernel $k(\cdot,\cdot)$. 
The squared maximum mean discrepancy between $P$ and $Q$ is defined as
\begin{equation}\label{eq:MMDoriginal}
(\mathrm{D}^k_2(P,Q))^2
:= \sup_{\|f\|_{\mathcal{H}} \leq 1}
\left( \int f \,\mathrm{d}P - \int f \,\mathrm{d}Q \right)^2,
\end{equation}
which can be written equivalently as
\[
\mathbb{E}\big[k(\mathbf{X},\mathbf{X}') + k(\mathbf{Y},\mathbf{Y}') - 2k(\mathbf{X},\mathbf{Y})\big],
\]
where $\mathbf{X},\mathbf{X}' \sim P$ and $\mathbf{Y},\mathbf{Y}' \sim Q$ are independent copies.
\end{definition}

The maximum mean discrepancy specializes to two cases that will be of particular interest in this text. 
The first corresponds to classical $L_2$ discrepancies in QMC, obtained by choosing a positive definite kernel $k:[0,1]^d \times [0,1]^d \rightarrow \mathbb{R}$. 
Without loss of generality, let $P$ denote the empirical distribution induced by a design set $P_N \subset [0,1]^d$, and let $Q$ denote the uniform distribution on the unit hypercube. 
Then, with a slight abuse of notation, the $L_2$ discrepancy of $P_N$ is defined as
\begin{equation}\label{eq:L2discrepancy}
(\mathrm{D}^k_2(P_N))^2
:= \iint_{[0,1]^d} k(\boldsymbol{x}, \boldsymbol{y}) \, \mathrm{d} \boldsymbol{x} \, \mathrm{d}\boldsymbol{y}
- \frac{2}{N} \sum_{i=1}^N \int_{[0,1]^d} k(\mathbf{X}_i, \boldsymbol{y}) \, \mathrm{d}\boldsymbol{y} 
+ \frac{1}{N^2} \sum_{i,j=1}^N k(\mathbf{X}_i, \mathbf{X}_j),
\end{equation}
for specific choices of $k$ yielding, for example, the star $D_2^*$, extreme $D^{\mathrm{ext}}_2$, and periodic $D_2^{\mathrm{per}}$ discrepancies.

For the second case, when both Borel measures $P$ and $Q$ are empirical, the maximum mean discrepancy reduces to a kernel-based discrepancy between two finite design sets, enabling direct comparison of samples drawn either from the same or from different underlying distributions. 
In this setting, again with a slight abuse of notation, for $P_N = \{\mathbf{X}_i\}_{i=1}^N\text{ and } P_m = \{\mathbf{Y}_i\}_{i=1}^m \subset [0,1]^d$, one may write \eqref{eq:MMDoriginal} as
\begin{equation}\label{eq:discreteMMD}
(\mathrm{D}^k_2(P_N, P_m))^2 
:= \frac{1}{N^2}\sum_{i,j=1}^N k(\mathbf{X}_i, \mathbf{X}_j)
-\frac{2}{mN} \sum_{i=1}^N \sum_{j=1}^m k(\mathbf{X}_i, \mathbf{Y}_j)
+ \frac{1}{m^2} \sum_{i,j=1}^m k(\mathbf{Y}_i, \mathbf{Y}_j).
\end{equation}

Together, \eqref{eq:star_disc}, \eqref{eq:L2discrepancy}, and \eqref{eq:discreteMMD} define the objective functions that we seek to minimize in this work.

\subsection{The Subset Selection Problem}

One particular computational LD construction strategy that has received considerable attention in recent years is the \emph{subset selection} procedure. It was originally formulated for the $L_\infty$ star discrepancy \eqref{eq:star_disc} in \cite{cle22, cle24_heuristic} and later generalized to kernel discrepancy measures in \cite{optimizingkerneldiscrepancy2025}. 

The idea is simple: given a large population of $N$ design points, one seeks a subset of size $m \ll N$ that is, in some sense, most representative of the full set. Representativeness will be measured in our setting by one of the discrepancy-based objectives introduced above.

\begin{center}
  \begin{minipage}{0.7\linewidth}
    \itshape
    Given a population set $P_N \subset [0,1]^d$ of size $N$, the goal is to select a subset $P_m \subset P_N$ of size $m \ll N$ that minimizes a discrepancy function, $D(\cdot)$. That is, we wish to solve the following optimization problem:
\begin{equation}\label{eq:optim_problem}
\underset{P_ m \subset P_N,\ |P_m| = m}{\textnormal{argmin}} D(P_m).
\end{equation}
  \end{minipage}
\end{center} 

Clearly, even for moderate values of $N$ and $m$, exhaustive search or exact methods for finding an optimum of this combinatorial problem are computationally infeasible. In fact, it was established in \cite[Theorem 3.1]{cle22} that \eqref{eq:optim_problem} is NP-hard when the discrepancy measure is the $L_\infty$ star discrepancy \eqref{eq:star_disc}. As a result, a range of heuristic approaches have been proposed for solving \eqref{eq:optim_problem}, most of which are based on swap-based local search strategies. We refer to \cite{clethesis24} for further reading on subset selection.

\subsection{Our Contribution}

In this work, we propose to address \eqref{eq:optim_problem} using Bayesian optimization. First, we generalize the hardness result and show in Section \ref{sec:NPHard} that subset selection with respect to a wide class of so-called kernel discrepancy measures is also, in fact, NP–hard. Next, we discuss the recently introduced framework of \emph{deep kernel embeddings} \cite{deepembeddings2020} to construct strictly positive definite covariance kernels for Gaussian Process surrogates over set-valued inputs which is key for our subset selection problem. Our full Bayesian optimization method is detailed in Section \ref{sec:BOforSubset}, and finally, numerical results illustrating its performance are presented in Section \ref{sec:results}.

\section{Kernel Subset Selection is NP--Hard}\label{sec:NPHard}

As discussed in the Introduction, subset selection is known to be computationally intractable in the classical $L_\infty$ star discrepancy setting \cite[Theorem 3.1]{cle22}. In particular, computing the $L_\infty$ star discrepancy itself is NP--hard \cite{starNPhard}, which strongly suggests that optimizing it over subsets is likewise intractable. This result provided theoretical justification for the widespread use of heuristic and approximate algorithms in that context.

We now extend this perspective to the kernel-based setting and show that, over the class of positive definite kernels, subset selection with respect to the maximum mean discrepancy is NP--hard in general. Specifically, we reduce from a known NP--hard setting, the constrained quadratic 0-1 problem \cite{MurtyKabadi1987}.

\begin{theorem}\label{thm:1}
For the inputs $P_N=\{\mathbf{X}_i\}_{i=1}^N\subset[0,1]^d$, an integer $m\le N$, a threshold $\tau>0$ and a positive definite kernel $k:[0,1]^d \times [0,1]^d \rightarrow \mathbb{R}$, the decision problem
\[
\exists\,P_m\subset P_N,\ |P_m|=m
\quad\text{such that}\quad
D_2^k(P_m)\le \tau
\]
is NP--hard. 
\end{theorem}

\begin{proof} From a population set $P_N = \{\mathbf{X}_i\}_{i=1}^N\subset[0,1]^d$ 
and an index subset $S \subset \{1,\ldots,N\}$ with $|S|=m$, let 
$P_m:=\{\mathbf{X}_i : i\in S\}$ and for a positive definite kernel 
$k:[0,1]^d\times[0,1]^d\to\mathbb R$, define
\[
D_2^k(P_m) = \sqrt{\iint k(\boldsymbol{x},\boldsymbol{y})\,d\boldsymbol{x}\,d\boldsymbol{y}
-\frac{2}{m}\sum_{i\in S}\int k(\mathbf{X}_i,\boldsymbol{y})\,d\boldsymbol{y}
+\frac{1}{m^2}\sum_{i,j\in S}k(\mathbf{X}_i,\mathbf{X}_j)}.
\]

\noindent
For such $k$, write
\[
C := \iint k(\boldsymbol{x},\boldsymbol{y})\,d\boldsymbol{x}\,d\boldsymbol{y},
\qquad
b_i := \int k(\mathbf{X}_i,\boldsymbol{y})\,d\boldsymbol{y},
\qquad
K_{ij}:=k(\mathbf{X}_i,\mathbf{X}_j).
\]
Let $\mathbf z\in\{0,1\}^N$ be the indicator of $S$, i.e., \ $z_i=\mathbf 1_{\{i\in S\}}$.
Then
\[
D_2^k(P_m)^2
= C - \frac{2}{m}\sum_{i=1}^N b_i z_i
  + \frac{1}{m^2}\sum_{i,j=1}^N K_{ij}z_i z_j.
\]
Multiplying by $m^2$ gives
\begin{equation}\label{eq:disc-quadratic}
m^2 D_2^k(P_m)^2
= m^2C -2m\,\mathbf b^\top \mathbf z + \mathbf z^\top K \mathbf z.
\end{equation}
\noindent
Hence, minimizing $D_2^k(P_m)$ over all $m$--element subsets is equivalent to minimizing
a quadratic polynomial in $\mathbf z$ under the constraint $\sum_i z_i=m$.

We prove the hardness result by reducing from the NP--hard decision problem known as the constrained quadratic 0-1 problem \cite{MurtyKabadi1987, GaoLi2013}: given a symmetric positive 
semidefinite matrix $Q\in\mathbb Q^{N\times N}$, vector 
$\mathbf c\in\mathbb Q^N$, integer $m$, and threshold $T$, decide whether
there exists $\mathbf z\in\{0,1\}^N$ with $\sum_i z_i=m$ such that
\begin{equation}
\mathbf z^\top Q \mathbf z + \mathbf c^\top \mathbf z \le T.
\label{eq:qubo}
\end{equation}
Given an instance $(Q,\mathbf c,m,T)$, define
\[
K := Q + I_N,
\qquad
\mathbf b := -\frac{1}{2m}\mathbf c.
\]
Since $Q\succeq 0$, we have $K\succ 0$. For any feasible $\mathbf z$ with $\sum_i z_i=m$,
\[
\mathbf z^\top K \mathbf z
= \mathbf z^\top Q \mathbf z + \mathbf z^\top \mathbf z
= \mathbf z^\top Q \mathbf z + m.
\]
Substituting into \eqref{eq:disc-quadratic} gives
\[
m^2 D_2^k(P_m)^2
= m^2 C + \mathbf z^\top Q \mathbf z + \mathbf c^\top \mathbf z + m.
\]
Thus
\[
\mathbf z^\top Q \mathbf z + \mathbf c^\top \mathbf z \le T
\quad\Longleftrightarrow\quad
m^2 D_2^k(P_m)^2 \le T + m + m^2 C.
\]
and by setting $\tau^2 := (T+m+m^2 C)/m^2$, we obtain equivalence of the two decision problems. 

It now remains to construct an positive definite kernel on $[0,1]^d$ realizing the quantities $(K,\mathbf b,C)$. Since $K\succ 0$, let $A$ be an invertible matrix with $A^\top A = K$.
Define
\[
\mu_f := A^{-T}\mathbf b,
\qquad
C := \|\mu_f\|_2^2 = \mathbf b^\top K^{-1}\mathbf b.
\]
We now construct a measurable feature map 
$\phi:[0,1]^d\to\mathbb R^N$ such that
\[
\phi(\mathbf X_i)=A \boldsymbol{e}_i
\quad\text{for } i=1,\dots,N,
\qquad
\int_{[0,1]^d}\phi(\boldsymbol{x})\,d\boldsymbol{x}=\mu_f.
\]
This can be done by defining $\phi$ to be piecewise constant on a partition
of $[0,1]^d$ with the desired integral, and redefining its values at the
finitely many points $\mathbf X_i$ (which does not affect the integral). Next, define
\[
k(\boldsymbol{x},\boldsymbol{y}):=\phi(\boldsymbol{x})^\top \phi(\boldsymbol{y}).
\]
Then $k$ is positive definite on $[0,1]^d$, and:

\begin{align*}
k(\mathbf X_i,\mathbf X_j) &=K_{ij}, \\
\int k(\mathbf X_i,\boldsymbol{y})\,d\boldsymbol{y} &= b_i, \\
\iint k(\boldsymbol{x},\boldsymbol{y})\,d\boldsymbol{x}\,d\boldsymbol{y}&=C.
\end{align*}

Thus the discrepancy computed from this $k$ satisfies 
\eqref{eq:disc-quadratic}. Finally, the NP--hard instance \eqref{eq:qubo} has a feasible solution
if and only if there exists a subset $P_m\subset P_N$
with $D_2^k(P_m)\le\tau$.\end{proof}

This resolves an open question posed by F. Clément \cite{clethesis24}.

\section{Subset Selection via Bayesian Optimization}

We now describe our Bayesian optimization approach for solving the subset selection problem \eqref{eq:optim_problem}. 
The key technical ingredient is a class of strictly positive definite kernels defined on finite subsets, known as \emph{deep embedding} kernels. 
We first review their construction and then describe the full optimization pipeline.

\subsection{Deep Embedding Kernels}
\label{sec:DEkernels}

Many learning and optimization problems involve set-valued inputs. 
Examples arise in natural language processing \cite{LLMsets2017}, clustering \cite{clustersets2019}, and in general experimental design. In our design setting, for a fixed population $P_N = \{\mathbf{X}_i\}_{i=1}^N \subset [0,1]^d$, let $\mathcal{S}_m(P_N)$ denote the set of all $m$-element subsets of $P_N$. Then the optimization variable is a subset $P_m \in \mathcal{S}_m(P_N)$, i.e., a set-valued input. 

Standard Gaussian process (GP) methodology is designed for Euclidean inputs and does not directly provide kernels on finite subsets. 
A classical construction for sets is the \emph{double-sum (DS) kernel} \cite{DS_intro_1999,DS_intro2_2002}, defined for subsets $P_m, \tilde{P}_m \in \mathcal{S}_m(P_N)$ by
\begin{equation}
\label{eq:DS_kernel}
k_0(P_m,\tilde P_m)
=
\frac{1}{|P_m|\,|\tilde P_m|}
\sum_{\mathbf{X} \in P_m}
\sum_{\mathbf{X}' \in \tilde P_m}
k_X(\mathbf{X},\mathbf{X}'),
\end{equation}
where $k_X : [0,1]^d \times [0,1]^d \to \mathbb{R}$ is a positive definite kernel acting on the base space. While $k_0$ is positive definite, it is generally not strictly positive definite, which may lead to problems downstream like singular Gram matrices in GP regression and optimization.

To address this, \cite{deepembeddings2020} introduced \emph{deep embedding (DE) kernels}. 
The construction begins by embedding each subset into the reproducing kernel Hilbert space (RKHS) $\mathcal{H}_{k_X}$ via the empirical kernel mean embedding
\[
\mathcal{E}(P_m)
=
\frac{1}{|P_m|}
\sum_{\mathbf{X} \in P_m}
k_X(\cdot,\mathbf{X}).
\]
The canonical RKHS distance between two subsets is then
\begin{equation}
\label{eq:RKHS_distance}
d_{\mathcal E}(P_m,\tilde P_m)
=
\|\mathcal{E}(P_m)-\mathcal{E}(\tilde P_m)\|_{\mathcal{H}_{k_X}}
=
\sqrt{
k_0(P_m,P_m)
+
k_0(\tilde P_m,\tilde P_m)
-
2k_0(P_m,\tilde P_m)
}.
\end{equation}
and a DE kernel is obtained by composing this distance with a radial positive definite kernel on Hilbert space:
\begin{equation}
\label{eq:DEkernel}
k_{\mathrm{DE}}(P_m,\tilde P_m)
=
k_H\!\left(
d_{\mathcal E}(P_m,\tilde P_m)
\right),
\end{equation}
where $k_H : [0,\infty) \to \mathbb{R}$ is such that $(h,h') \mapsto k_H(\|h-h'\|_{\mathcal H})$ is positive definite for any Hilbert space $\mathcal H$. Under mild conditions, $k_{\mathrm{DE}}$ is in fact \emph{strictly} positive definite acting on the set of all subsets \cite[Propositions 1–3]{deepembeddings2020}.

\subsection{Bayesian Optimization for Subset Selection}
\label{sec:BOforSubset}

We now describe our BO framework over the discrete domain $\mathcal{S}_m(P_N)$. Let $k$ be a positive definite kernel on $[0,1]^d$ and define the objective
\[
f(P_m)
=
D_2^{k}(P_m),
\qquad
P_m \in \mathcal{S}_m(P_N).
\]
In practice, to improve numerical stability, we model the logarithm of the discrepancy,
\[
f(P_m)
=
\log D_2^{k}(P_m).
\]
We place a zero-mean Gaussian process prior on $f$ with covariance kernel given by a DE kernel of the form
\[
k_{\mathrm{DE}}(P_m,\tilde P_m)
=
\exp\!\left(
-
\frac{
d_{\mathcal E}(P_m,\tilde P_m)^2
}{
2\theta_H^2
}
\right),
\]
where $d_{\mathcal E}$ is defined in \eqref{eq:RKHS_distance}. 
At the point level we use a radial basis function (RBF) kernel $k_X(\boldsymbol{x},\boldsymbol{x}')
=
\exp\!\left(
-\|\boldsymbol{x}-\boldsymbol{x}'\|^2/2\sigma_X^2
\right)$.
The hyperparameters $(\sigma_X,\theta_H)$ are selected at each BO iteration by maximizing the GP log marginal likelihood over log-spaced grids, after standardizing the observed targets to zero mean and unit variance.

Let $\{P_m^{(i)}\}_{i=1}^t$ denote the subsets evaluated up to iteration $t$, and define
\[
f_{\min}
=
\min_{1 \le i \le t} f(P_m^{(i)}).
\]
For minimization, we employ the Expected Improvement acquisition function
\[
\mathrm{EI}(P_m)
=
\mathbb{E}\!\left[
(f_{\min}-f(P_m))_{+}
\,\middle|\,
\mathcal{D}_t
\right],
\]
which admits the standard closed-form expression in terms of the GP posterior mean $\mu(P_m)$ and variance $\sigma^2(P_m)$:
\[
\mathrm{EI}(P_m)
=
(f_{\min}-\mu(P_m))
\Phi\!\left(
\frac{f_{\min}-\mu(P_m)}{\sigma(P_m)}
\right)
+
\sigma(P_m)
\phi\!\left(
\frac{f_{\min}-\mu(P_m)}{\sigma(P_m)}
\right),
\]
where $\Phi$ and $\phi$ denote the standard normal CDF and PDF.

Since direct maximization of EI over the combinatorial domain $\mathcal{S}_m(P_N)$ is, in itself, infeasible, we approximate the maximization via local search over \emph{1-swap neighborhoods}. 
Given a subset $P_m$, a 1-swap neighbor is any subset of the form
\[
\bar P_m
=
P_m \setminus \{\mathbf{X}_i\}
\cup
\{\mathbf{X}_j\},
\qquad
\mathbf{X}_i \in P_m,
\quad
\mathbf{X}_j \in P_N \setminus P_m.
\]

\begin{algorithm}
\caption{BO-DE for Low-Discrepancy Subset Selection}
\label{alg:bo-de}
\begin{algorithmic}[1]
\State \textbf{Input:} $P_N=\{\mathbf{X}_i\}_{i=1}^N$, subset size $m$, discrepancy $D(\cdot)$, $n_{\mathrm{init}}$, iterations $T$, grids $\Sigma_X$ and $\Theta_H$, hill-climb steps $S_{\max}$, neighbors per step $M$, restarts $R$
\State \textbf{Output:} Best subset $P_m^\star\in\mathcal S_m(P_N)$ found

\State Sample $P_m^{(1)},\dots,P_m^{(n_{\mathrm{init}})}\sim\mathrm{Unif}(\mathcal S_m(P_N))$
\State Evaluate $y_i \gets D(P_m^{(i)})$ and set $\tilde y_i \gets \log(\max\{y_i,\varepsilon\})$ \emph{(or $\tilde y_i\gets y_i$)}
\State $\mathcal D \gets \{(P_m^{(i)},\tilde y_i)\}_{i=1}^{n_{\mathrm{init}}}$

\For{$t=1$ \textbf{to} $T$}
  \State Choose $(\sigma_X,\theta_H)\in\Sigma_X\times\Theta_H$ maximizing GP log marginal likelihood on standardized targets
  \State Fit GP surrogate on $\mathcal D$ with $k_{\mathrm{DE}}^{(\sigma_X,\theta_H)}$;
  \State \hspace{1em} let $f_{\min}=\min_{(P_m,\tilde y)\in\mathcal D}\tilde y$ and $P_m^{\mathrm{best}}=\arg\min_{(P_m,\tilde y)\in\mathcal D}\tilde y$
  \State Run EI hill-climb from $P_m^{\mathrm{best}}$ and from $R$ random starts;
  \State \hspace{1em} each step samples $\le M$ 1-swap neighbors and moves if EI increases; collect endpoints $\mathcal C$
  \State Select $P_m^{\mathrm{next}}=\arg\max_{P_m\in\mathcal C}\mathrm{EI}(P_m)$
  \State Evaluate $y^{\mathrm{next}} \gets D(P_m^{\mathrm{next}})$ and transform $\tilde y^{\mathrm{next}}$;
  \State \hspace{1em} append $(P_m^{\mathrm{next}},\tilde y^{\mathrm{next}})$ to $\mathcal D$
\EndFor

\State \textbf{return} $P_m^\star=\arg\min_{(P_m,\tilde y)\in\mathcal D}\tilde y$
\State \hspace{1em} \emph{(equivalently $\arg\min D(P_m)$ when $\tilde y=\log D(P_m)$).}
\end{algorithmic}
\end{algorithm}

Starting from an initial subset, we sample a finite collection of 1-swap neighbors and move to a neighbor with strictly larger EI whenever one exists. 
The procedure terminates at a local maximizer of EI. 
To encourage exploration, each BO iteration performs one hill-climb initialized at the current best subset and several additional hill-climbs initialized from randomly sampled subsets. 
The next evaluation point is selected as the local maximizer with the largest EI among these runs.

\section{Numerical Experiments}\label{sec:results}

We now evaluate the proposed Bayesian optimization framework for subset selection. Following the Introduction, we consider three objectives which cover several related settings selecting an $m$–point design $P_m$ from a larger population $P_N\subset[0,1]^d$ to approximate a reference measure, which may be discrete or continuous. When the reference measure is discrete, \eqref{eq:discreteMMD} will be used as the objective function and \eqref{eq:star_disc} or \eqref{eq:L2discrepancy} otherwise. In the continuous case, we are interested in optimizing the classical $L_\infty$ star discrepancy, on which the subset selection problem was originally framed in the QMC literature, and also the $L_2$-type kernel discrepancies.

\subsection{Experimental Setup}

Our goal is to achieve substantial discrepancy minimization using a limited number of true objective evaluations. For this reason, in all experiments we restrict the evaluation budget to $100$ discrepancy evaluations: $50$ initial evaluations used to fit the Gaussian Process surrogate, followed by $50$ additional evaluations selected sequentially by maximizing the EI acquisition function.

We focus on two-dimensional settings throughout and compare BO-DE against three baseline methods:

\begin{itemize}
    \item \textbf{Random:} Subsets are sampled uniformly at random from $\mathcal S_m(P_N)$ and evaluated using the true discrepancy.
    
    \item \textbf{BO-DS:} This method replaces the DE kernel with the double-sum (DS) kernel, with the same fixed RBF base kernel. Acquisition optimization is performed using the same EI hill-climbing procedure over sampled 1-swap neighborhoods as in the DE method.
    
    \item \textbf{Greedy Local Swap (GLS):} This method directly optimizes the discrepancy without a surrogate model. Starting from a collection of randomly sampled subsets, the best initial subset is selected and iteratively improved via 1-swap moves. At each step, a sampled collection of 1-swap neighbors is evaluated using the true discrepancy, and the algorithm moves to the best improving neighbor. If no improvement is found, a random restart is performed until the evaluation budget is exhausted.
\end{itemize}

All methods are compared under equal discrepancy evaluations budgets rather than wall-clock time.

\subsection{Results}

\textbf{Experiment 1.} We minimize the symmetric $L_2$ discrepancy \eqref{eq:L2discrepancy} induced by the product kernel $k(\boldsymbol{x},\boldsymbol{x}')=\prod_{j=1}^d (1 - 2|x_j - x'_j|)/4$. This choice corresponds to a classical $L_2$ discrepancy widely used in quasi-Monte Carlo methods, though other $L_2$ variants could equally be considered within our framework.
A population $P_N$ of size $N=1000$ is sampled uniformly from $[0,1]^2$, and subsets of size $m=25$ are selected.

Figure~\ref{fig:symmetric_results} reports the best discrepancy value observed as a function of the evaluation count. BO-DE achieves the smallest final discrepancy and produces the largest reduction relative to the initial random design. GLS yields the second-best performance, improving upon the initial subsets but plateauing once no improving swaps are found. In contrast, both BO-DS and random sampling show no improvement beyond the initial $50$ random evaluations, indicating limited effectiveness in this budget-restricting setting.

\begin{figure}[h!]
    \centering
    \includegraphics[width=0.66\linewidth]{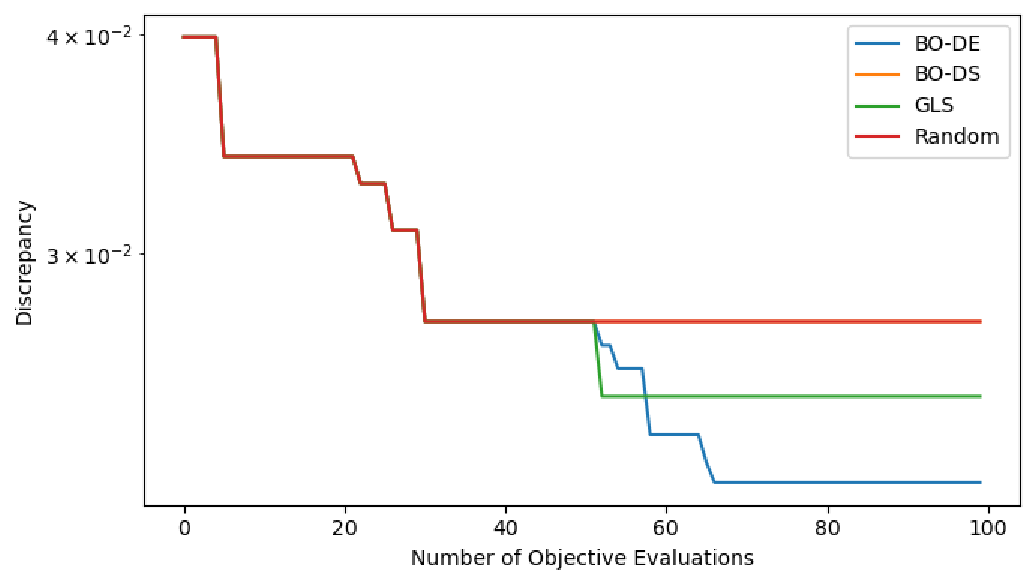}
    \caption{Symmetric discrepancy minimization for $N=1000$ and $m=25$.}
    \label{fig:symmetric_results}
\end{figure}

\noindent
\textbf{Experiment 2.} We next consider the maximum mean discrepancy as defined in \eqref{eq:discreteMMD}. A population of size $N=1000$ is sampled from a two-component Gaussian mixture, and a subset of size $m=25$ is selected under the same evaluation budget as in the previous experiments.

The results are shown in Figure~\ref{fig:MMDgaussian}. Consistent with the earlier discrepancy experiments, BO-DE attains the smallest final discrepancy value within the fixed budget. GLS yields moderate improvement, while both random sampling and BO with the DS kernel show limited reduction beyond the initial designs.

For visual comparison, the final subsets obtained by GLS and BO-DE are displayed as black points over the original population (shown in grey). The BO-DE subset exhibits more uniform coverage of the mixture components, with fewer visible clusters and sparse regions than the GLS solution.

\begin{figure}[t]
  \centering
  \includegraphics[width=0.32\textwidth]{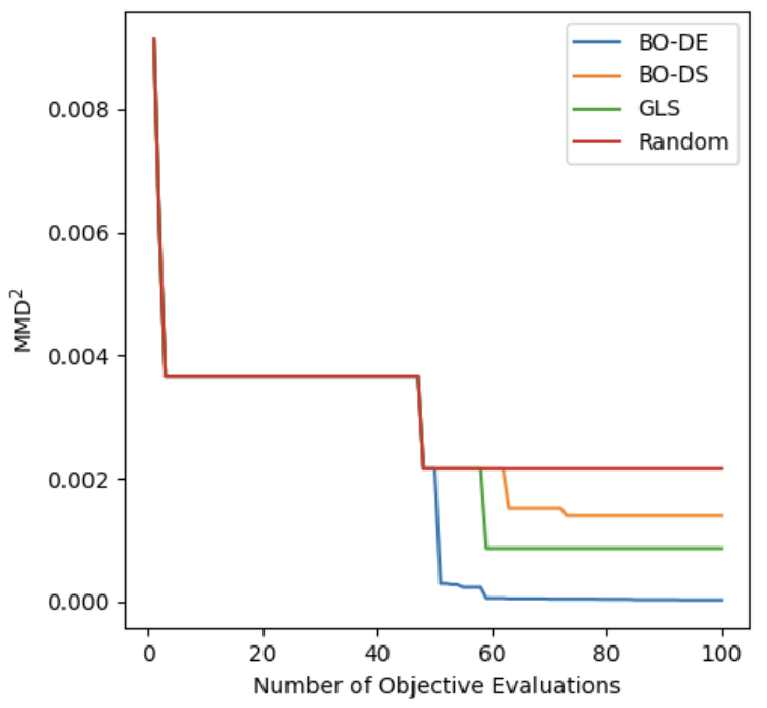}
  \includegraphics[width=0.32\textwidth]{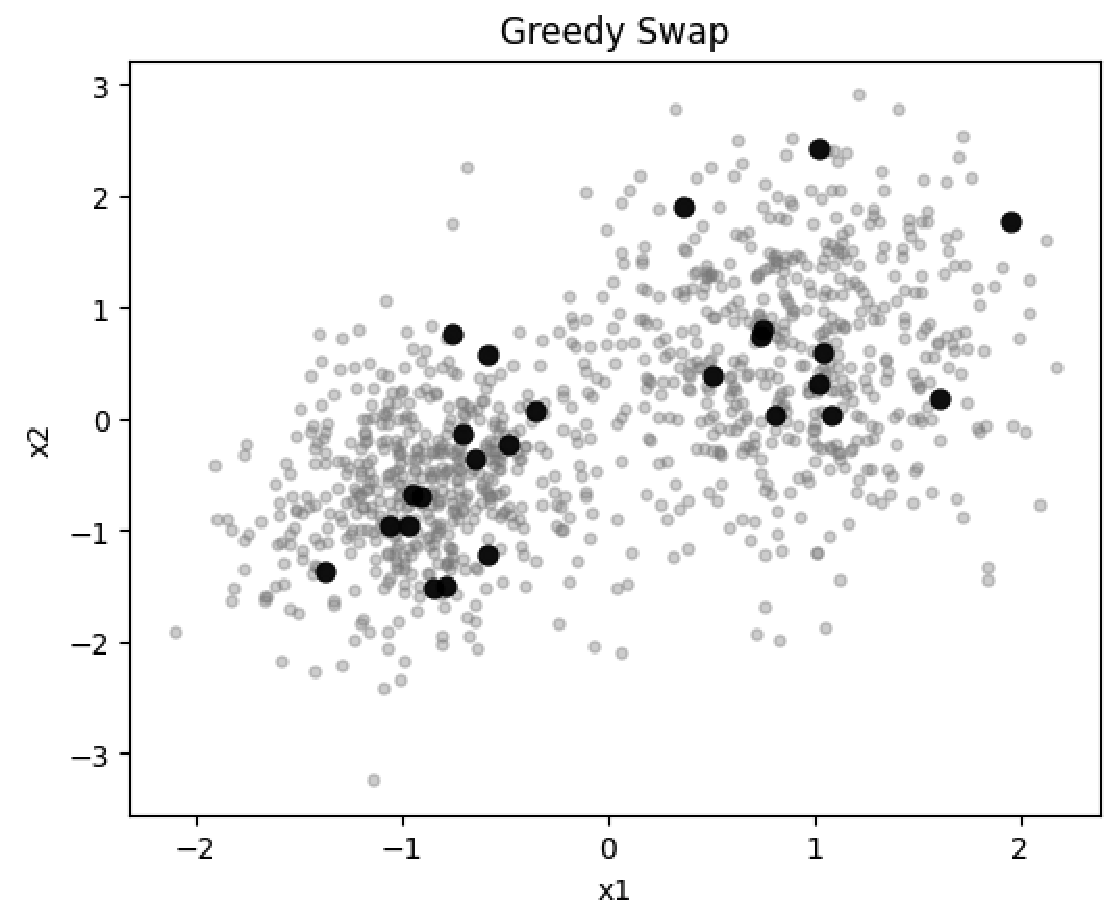}
  \includegraphics[width=0.32\textwidth]{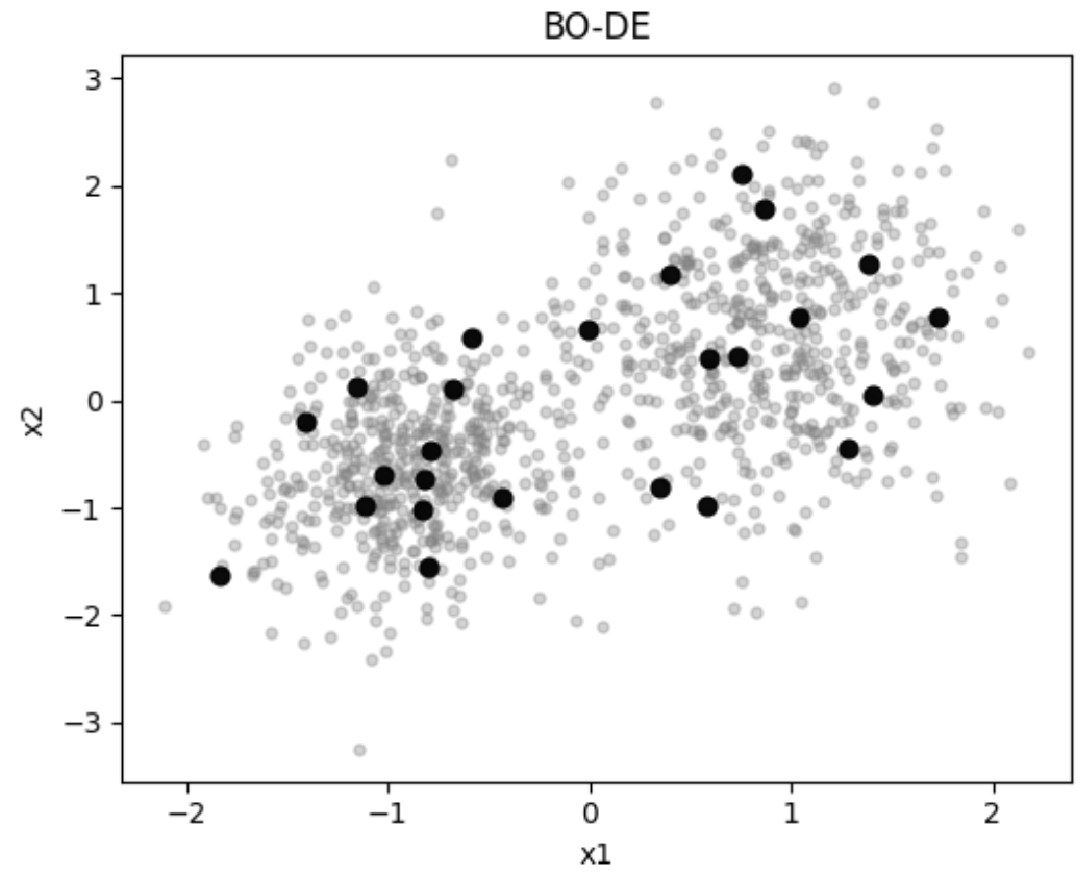}
  \caption{Maximum mean discrepancy minimization for two-component Gaussian mixture for random, GLS, BO-DS and BO-DE methods for $N=1000$ and $m=25$. The resulting subset for GLS (\textbf{Middle}) and BO-DE (\textbf{Right}).}\label{fig:MMDgaussian}
\end{figure}

\begin{figure}[h!]
    \centering
    \includegraphics[width=0.66\linewidth]{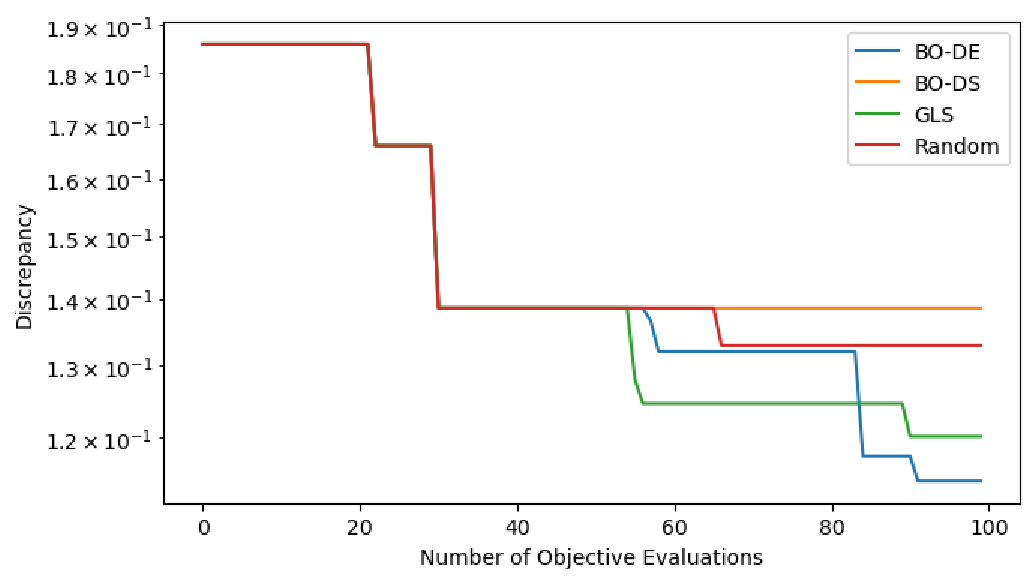}
    \caption{$L_\infty$ star discrepancy minimization for $N=1000$ and $m=25$.}
    \label{fig:star_results}
\end{figure}

\noindent
\textbf{Experiment 3.} Lastly, we consider minimization of the $L_\infty$ star discrepancy \eqref{eq:star_disc} under the same experimental protocol. The $L_\infty$ star discrepancy is well known to be computationally demanding---even the fastest exact algorithms scale on the order of $\mathcal{O}(N^{d/2+1})$; see \cite{DEM}. Consequently, evaluating the true objective becomes increasingly expensive as $N$ grows, making surrogate-based approaches particularly appealing in this setting.

Figure~\ref{fig:star_results} shows the best star discrepancy values. As in both previous experiments, BO-DE achieves the smallest final discrepancy. GLS provides moderate improvement but both BO-DS and uniform random sampling exhibit little to no improvement once again beyond the initial $50$ random evaluations. These results suggest that the DE surrogate is especially beneficial when the objective function is expensive to compute, as is the case for the $L_\infty$ star discrepancy.

\section{Discussion}

In this work, we introduced a Bayesian Optimization framework for selecting low-discrepancy subsets based on several variants of the discrepancy, using deep embedding kernels to construct Gaussian Process surrogates directly over the combinatorial space of subsets. Analysis first established that the subset selection problem with respect to kernel discrepancies is NP-hard in general and we demonstrated the efficacy of the method empirically illustrating that the DE-based surrogate can substantially improve discrepancy reduction with small numbers of true objective evaluations.

Although our experiments focused on standard geometric discrepancies, the methodology is considerably more general. This paper further emphasizes that deep embedding kernels provide a principled way to perform black-box optimization over structured combinatorial domains, an important setting in many modern statistical and scientific domains.

\subsubsection*{Acknowledgments.} This work was partially supported by the U.S. National Science Foundation (NSF DMS–2316011). The author thanks Fred J. Hickernell for the invitation to speak at the 2025 IMSI Workshop ``Kernel Methods in Uncertainty Quantification and Experimental Design", where the initial idea for this paper emerged in conversation with David Ginsbourger. The author also thanks Illinois Institute of Technology student Sanjana Waghray for assistance with exploratory numerical experiments.

\bibliography{sn-bibliography}

\end{document}